%% file: robust-tree.tex
\documentclass{llncs}

\usepackage{amsmath}
\usepackage{graphicx}
\usepackage{latexsym}
\usepackage{theorem}
\usepackage{subfig}
\usepackage{graphicx}
\usepackage{float}
\usepackage{algorithmic}
\usepackage[ruled]{algorithm}
\usepackage{citesort}

\makeatletter
    
    \@addtoreset{algorithm}{section}
\makeatother

\theoremstyle{plain}
\theorembodyfont{\rm}
\newtheorem{fact}{Fact}

\newcommand{\Omit}[1]{}

\newcounter{Codeline}

\newcommand{\nodeboundary}[1]{\partial {#1}}

\begin{document}

\title{Physical Expander in Virtual Tree Overlay\thanks{This work is 
supported in part by Grand-in-Aid for Young 
Scientists ((B)22700010) of JSPS. Additional support from ANR projects R-Discover, SHAMAN, and ALADDIN.
}}

\titlerunning{Physical Expander in Virtual Tree Overlay}

\author{Taisuke Izumi\inst{1} \and Maria Gradinariu Potop-Butucaru\inst{2}
\and \\ Mathieu Valero\inst{2}}
        
\authorrunning{T. Izumi, M. Gradinariu Potop-Butucaru, and S. Tixeuil}

\institute{
Graduate School of Engineering, Nagoya Institute of Technology, Japan\\
\email{t-izumi@nitech.ac.jp}
\and
Universit\'{e} Pierre et Marie Curie - Paris 6, LIP6 CNRS 7606, France\\
\email{maria.gradinariu@lip6.fr, mathieu.valero@gmail.com}
}
\maketitle

\begin{abstract}
In this paper, we propose a new construction
of constant-degree expanders motivated by their application in P2P overlay networks 
and in particular in the design of robust trees overlay. 

Our key result can be stated as follows.
Consider a complete binary tree $T$ and construct a random pairing 
$\Pi$ between leaf nodes and internal nodes. We prove that the graph $G_\Pi$ 
obtained from $T$  by contracting all pairs (leaf-internal nodes) achieves a constant 
node expansion with high probability. 
The use of our result in improving 
the robustness of tree overlays is straightforward. That is,  
if each physical node participating to the overlay manages a 
random pair that couples one virtual internal node and one virtual 
leaf node then the physical-node 
layer exhibits a constant expansion with high probability. 
We encompass the difficulty of obtaining this random tree virtualization by 
proposing a local, self-organizing and churn resilient uniformly-random 
pairing algorithm with $O(\log^2 n)$ running time. 
Our algorithm has the merit 
to not modify the original tree virtual overlay 
(we just control the mapping between physical nodes and virtual nodes). 
Therefore, our scheme is general and can be applied to a large number of
tree overlay implementations. We validate its 
performances in dynamic environments via extensive simulations. 
\end{abstract}

\input{introduction}

\input{related}

\input{model}
\input{main-result}
\input{mixing}

\input{simulations}
\input{conclusion}

\bibliographystyle{abbrv}
\bibliography{books,bib}
\include{appendix}

\appendix

\end{document}

%% file: introduction.tex
\section{Introduction}

\paragraph{Background and Motivation}
\sloppy{
P2P networks are appealing for sharing/diffusing/searching resources among 
heterogeneous groups of users. Efficiently organizing users in order 
to achieve these goals is the main concern that motivated the study 
of overlay networks. In particular, {\em tree overlays} recently 
becomes an attractive class of overlay networks because efficient 
implementations of various communication primitives in P2P systems tied to 
the hierarchical and acyclic properties of trees such as content-based 
publish/subscribe or multicast (\cite{CK09,1009738,EGHKK03,CDKNRS03}). 
Many P2P and distributed variants of classical tree 
structures such as B-trees, R-trees or P-trees have been designed so far
\cite{BATON,VBI,PGrid,Caron06,DRTree,Brushwood}. 
}

Because of the dynamic nature of P2P networks, robustness 
to faults and churn (e.g., frequent node join and leave) is 
indispensable for the service on the top of them to function properly. 
A recent trend in measuring the robustness of overlay networks is the 
evaluation of {\em graph expansion}. The (node) expansion $h(G)$ 
of an undirected graph $G = (V_G, E_G)$ is 
defined as:
\[
h(G) = \min_{S \subseteq V_G | |S| \leq n/2} \frac{|\nodeboundary{S}|}{|S|},
\]
where $\nodeboundary{S}$ is the set of nodes that are adjacent to a node in 
$S$ but not contained in $S$. The implication of node expansion is that 
the deletion of at least $h(G) \cdot k$ nodes is necessary to disconnect a 
component of $k$ nodes in $G$. That is, graphs with good
expansion are hard to be partitioned into a number of large connected 
components.  
In this sense, the expansion of a graph can be seen as a good evaluation 
of its resilience to faults and churn. Interestingly, the expansion of 
tree overlays is trivially $O(1/n)$, which is far from adequate. 
This weakness to faults is the primary reason why tree overlays 
are not pervasive in real applications. 

Our focus in this paper is to provide the mechanism to make
tree overlays robust and suitable to real applications. 
In particular, we are interested in {\em generic}
schemes applicable to a large class of tree overlays with minimal extra cost: 
As seen above, there are many variations of tree-based data structures with 
distinguished characteristics, but their distributed implementations 
always face the problem how to circumvent the threat of disconnection.
Therefore, providing such a generic robustization scheme 
would offer the substantial benefit for implementing distributed
tree-based data structures in a systematic way.

\paragraph{Our contribution} 
Solutions for featuring P2P tree overlays with robustness range from increasing 
the connectivity of the overlay (in order to  eventually mask the network churn
and fault) to adding additional mechanism for monitoring and repairing the
overlay. However the efficiency of these techniques is shadowed by the
extra-cost needed to their implementation in dynamic settings. Moreover,
the design of those mechanisms often depends on some specific tree 
overlay implementation, and thus their generalization is difficult.
Therefore, we propose a totally novel approach that exploits the 
principal of {\em tree virtualization}. That is, in a tree overlay 
one physical node may be in charge of several virtual nodes. 
The core of our approach is to use this mapping between virtual and 
physical nodes such that the physical-node layer exhibits a good 
robustness property. 

Our primary contribution is the following theorem 
which is the key in the construction of our random virtualization scheme:
\begin{theorem} \label{thmExpansion}
Let $T$ be a complete $n$-node binary tree with duplication of the 
root node (the duplicated root is seen to be identical to
the original root). Then, we can define a bijective function 
$\Pi$ from leaf nodes to internal nodes. Let $G_\Pi$ be the 
graph obtained from $T$ by contracting pair $(v, \Pi(v))$ for 
all $v$ \footnote{The contraction of $(v, \Pi(v))$ means that we 
contract edge $\{v, \Pi(v)\}$ as if it exists in $T$.}. 
Choosing $\Pi$ uniformly at random $G_\Pi$ has a constant (node) 
expansion with high probability.
\end{theorem}

An immediate consequence of this theorem is that the physical-node 
layer achieves a constant expansion with high probability
if a random chosen couple composed of one leaf and one internal node 
is assigned to each physical node. It should be noted that our random 
tree virtualization does not modify the original properties of 
the tree overlay since we only control the mapping to physical nodes.
This feature yields a general applicability of our result to a large
class of tree overlay implementations.

The above result relies on the uniform random bijection (i.e.
random perfect bipartite matching) between 
internal and leaf nodes in the tree overlay. Therefore, 
in order to prove the effectiveness of our proposal in a P2P context 
we also addressed the construction of random perfect bipartite matching
over internal and leaf nodes. Interestingly, we can propose a 
local and self-organizing scheme based on the parallel random 
permutation algorithm by Czumaj et. al.\cite{CK00}. Our scheme 
allows us to increase the graph expansion to a constant within 
$O(\log^2 n)$ synchronous rounds ($n$ is the number of physical nodes). 
The quick convergence of our scheme in dynamic settings is validated 
through extended simulations. 

\paragraph{Roadmap}
In Section \ref{secRelated}, we introduce the relate work mainly in the field
of distributed computing. Section \ref{secMain} presents the proof of our
main result. The issue about the distributed implementation of our scheme
is explained in Section \ref{secDistributed} which includes the simulation 
result. Finally, Section 
\ref{secConclusion} provides the conclusion and future research issues.

%% file: related.tex
\section{Related works} \label{secRelated}
Expander graphs have been studied extensively in 
many areas of theoretical computer science. A good tutorial can be 
found in \cite{HLW06}. In the following we restrict our attention 
to distributed constructions with a special emphasize on specific 
P2P design. 
   
There are several results about expander construction in distributed
settings. Most of those results are based on the 
distributed construction of random regular graphs, which exhibit
a good expansion with high probability. To the best of our knowledge 
one of the first papers that addressed expander constructions 
in peer-to-peer settings is \cite{LS03}. The authors compose $d$ 
Hamiltonian cycles to obtain a $2d$-regular graph. In \cite{RSW05} 
the authors propose a fault-tolerant expander 
construction using a pre-constructed tree overlay. 
It provides the mechanism to maintain an approximate
random $d$-regular graph under the assumption that the system always 
manages a spanning tree. The distributed construction of random 
regular graphs based on a stochastic graph transformation is also 
considered in \cite{FGMS06,CDH09}. They prove that repeating a specific
stochastic graph modification (e.g., swapping the two endpoints 
of a length-three path) eventually returns a uniformly-random sampling 
of regular graphs. Since all the previously mentioned algorithms 
are specialized in providing good expansion, the combination 
with overlays maintenance is out of their scope. Therefore, these 
works cannot be easily extended to a generic fault tolerant mechanism 
in order to improve the resiliency of a distributed overlay. 
Contrary to the previous mentioned works, our study can be seen as 
a way of identifying implicit expander properties in a given 
topological structure. There are several works along this direction. 
In \cite{DT10} the authors propose a self-stabilizing constructions 
of spanders, which are spanning subgraphs including smaller number of edges 
than the original graph but having the asymptotically same expansion
as the original\footnote{A spander is also called a {\em sparsifier}.}. 
Abraham et.al. \cite{AAY05} and Aspnes and Wieder \cite{AW08} respectively 
give the analysis of the expansion for some specific distributed data 
structures (skip graphs and skip b-trees). Recently Goyal et.al. 
\cite{GRV09} prove that given a graph $G$, the composition of two 
random spanning trees has the expansion at least $\Omega(h(G)/\log n)$, 
where $n$ is the number of node in $G$. We can differentiate our
result from the above works by its generality and the novelty of 
random tree virtualization concept.

%% file: main-result.tex
\section{The expander property of $G_\Pi$} 
\label{secMain}

\subsection{Notations}
For an undirected graph $G$, $V_G$ and $E_G$ respectively denote the sets 
of all nodes and edges in $G$. Given a graph $G$ and a subset of nodes 
$S \subseteq V_G$, we define $\mathrm{Ind(S)}$ to be 
the subgraph of $G$ induced by $S$. 
For a set of nodes $S$, its complement
is denoted by $\overline{S}$. The {\em node boundary} of 
a set $S \subseteq V_G$ is defined as a set of nodes in $\overline{S}$ that
connect to at least one node in $S$, which is denoted by $\nodeboundary{S}$.

Let $T = (V_T, E_T)$ be a binary tree. The sets of leaf nodes and 
internal nodes for $T$ are respectively denoted by $L(V_T)$ and $I(V_T)$. 
Given a subset $S \subseteq V_T$, we also define $L(S) = L(V_T) \cap S$ 
and $I(S) = I(V_T) \cap S$. For a (sub)tree $X$, the root node of $X$ 
is denoted by $r(X)$, and the parent of $r(X)$ is denoted by $p(X)$.
 
\subsection{Preliminary results}
In the following $T$ denotes  a complete binary tree without 
explicit statement. The root of $T$ is denoted by $r(T)$. 
We prove several auxiliary results that will be further 
used in our main result.

\begin{lemma} \label{lmaConnectedBoundarySize}
For any nonempty subset $S \subseteq I(V_T)$ such that
$\mathrm{Ind}(S)$ is connected, $|\nodeboundary{S}| \geq |S| + 1$.
In particular, if $r(T) \not\in S$ holds, 
$|\nodeboundary{S}| \geq |S| + 2$.  
\end{lemma}

The above lemma can be generalized for any (possibly disconnected) 
subset $S \subseteq I(V_T)$.

\begin{lemma} \label{lmaGeneralBoundarySize}
Given any nonempty subset $S \subseteq I(V_T) \setminus \{r(T)\}$ such that
$\mathrm{Ind}(S)$ of $T$ consists of $m$ connected components, 
$|\nodeboundary{S}| \geq |S| + m + 1$ holds. 
\end{lemma}

The following corollary is simply deduced from Lemma \ref{lmaGeneralBoundarySize}.

\begin{corollary} \label{corolBoundarySize}
Let $X$ be a subtree of $T$. For any subset $S \subseteq I(V_X)$,
$|\nodeboundary{S} \cap V_X| \geq |S \cap V_X|$. 
In particular, if $S$ is nonempty, we have 
$|\nodeboundary{S} \cap V_X| \geq |S \cap V_X| + 1$.
\end{corollary}

\subsection{Main Result}

In what follows, $|L(V_T)|$ is denoted by $n$ for short (i.e., $n$ is 
the number of nodes in $G_\Pi$). We also assume $\Pi$ is a bijective
function from leaf nodes to the set of internal nodes (where the root 
doubly appears), which is chosen from all $n!$ possible functions
uniformly at random. For a subset of nodes 
$S \subseteq L(V_T)$, we define $\Pi(S) = \{\Pi(u) | u \in S\}$ and 
$Q_{\Pi} = S \cup \Pi(S)$. 

Provided a subset $S \subseteq L(V_T)$ satisfying 
$|S| < n/2$, we say a subtree $X$ is {\em $S$-occupied} 
if all of its leaf nodes belong to $S$.  
A $S$-occupied subtree $X$ is {\em maximal} if 
there is no $S$-occupied subtree $X$ containing $X$ as a subtree. 
Note that two $S$-occupied maximal subtrees $X_1$ and $X_2$ in a common tree $T$
are mutually disjoint and $p(X_1) \neq p(X_2)$ holds because of their maximality.
We first show two lemmas used in the main proof.

\begin{lemma} \label{lmaInternalBoundary}
Let $X$ be a maximal $S$-occupied subtree for a nonempty subset 
$S \subseteq L(V_T)$. Then, $|(\nodeboundary{Q_{\Pi}}) \cap V_X| 
\geq |\overline{Q_{\Pi}} \cap V_{X}|/2$ holds. 
\end{lemma}

\begin{lemma} \label{lmaContainmentProb}
Given a subset $S \subseteq L(V_T)$ such that $|S| \leq n/2$, 
let $X_0, X_1, \cdots X_k$ be all maximal $S$-occupied subtrees and 
$V_X = \cup_{i = 1}^{k} V_{X_i}$.  For any $\alpha < 1$, 
$\Pr(|\Pi(S) \cap I(V_X)| \geq \alpha|I(V_X)|) \leq 
{|S| \choose \alpha(|S| - k)}\left(\frac{(|S| - k)}{n}\right)^{\alpha(|S| - k)}$.
\end{lemma}

The implication of the above two lemmas is stated as follows: 
We are focusing on a subset of boundary nodes $\nodeboundary{Q_\Pi}$ 
that are associated with some ``hole'' (that is, the set of nodes not 
contained in $Q_\Pi$) in $S$-occupied subtrees. Lemma 
\ref{lmaInternalBoundary} implies that at least half of the nodes 
organizing the hole belongs to $\nodeboundary{Q_\Pi}$. Lemma 
\ref{lmaContainmentProb} bounds the probability that $S$-occupied 
subtrees has the hole with size larger or equal to $(1 - \alpha)|I(V_X)|$.
We also use the following inequality:

\begin{fact}[Jensen's inequality] \label{thmJensen}
Let $f$ be the convex function, $p_1, p_2, \cdots$ be a series of
real values satisfying
$\sum_{i=1}^{\infty} p_i = 1$, and $x_1, x_2, \cdots$ be a series
of real values. Then, the following inequality holds:
\[
\sum_{i=1}^{\infty} p_i f(x_i) \geq f(\sum_{i=1}^{\infty} p_i x_i)
\]
\end{fact}

We give the proof of the main theorem (the statement is refined).

\addtocounter{theorem}{-1}
\begin{theorem} 
The node expansion of $G_\Pi$ is at least $\frac{1}{480}$ with 
probability $1 - o(1)$.
\end{theorem}

\begin{proof}
To prove the lemma, we show that with high probability, 
$|\nodeboundary{S}| \geq |S|/480$ holds for any subset 
$S \subseteq V_{G_\Pi}$ such that $|S| \leq n/2$. In the proof, 
we identify $L(T)$ and $V_{G_\Pi}$ as the same set, and thus 
we often refer $S$ as a subset of $L(V_T)$ without explicit 
statement. Given a set $S$, let $k$ be the number of maximal 
$S$-occupied subtrees, $\mathcal{X} = \{X_1, X_2, \cdots, X_{k}\}$ 
be all maximal $S$-occupied trees, and $V_X = \cup_{i = 1}^{k} V_{X_i}$. 
We also define $P = \{p(X_i) | X_i \in \mathcal{X}\}$. 
Throughout this proof, we omit the subscript $\Pi$ of $Q_{\Pi}$.
For a subset $Y \subseteq V_T$, let $Q(Y) = Q \cap Y$ and 
$q(Y) = |Q(Y)|$ for short.

The goal of this proof is to show that $|\nodeboundary{Q}| \geq |S|/240$ 
holds with high probability in $T$ for any given $S$. 
It follows that $|\nodeboundary{S}| \geq |S|/480$ holds
in $G_{\Pi}$. To bound $|\nodeboundary{Q}|$, we first consider 
the cardinality of two subsets $(\nodeboundary{Q}) \cap V_X$ and 
$\nodeboundary{(Q \cap \overline{V_X})}$ 
$(= \nodeboundary{Q(\overline{V_X})})$. See Figure \ref{fig:tree-Boundary} for an example.
While these subsets are not mutually disjoint, only the roots of 
$S$-occupied subtrees can be contained in
$\nodeboundary{Q(\overline{V_X})}$.  It implies
\begin{eqnarray}  
|(\nodeboundary{Q} \cap V_X) \cap 
\nodeboundary{Q(\overline{V_X})}| &\leq& q(P). \label{mathIntersection}
\end{eqnarray} 

Lemma \ref{lmaGeneralBoundarySize} and \ref{lmaInternalBoundary} lead 
to the following inequalities:
\begin{eqnarray}
|(\nodeboundary{Q}) \cap V_X| &\geq& 
\frac{1}{2} \left(\sum_{i = 1}^{k} |\overline{Q} \cap 
V_{X_i}| \right) \nonumber \\
&\geq& (|S| - k - q(V_X))/2. \label{mathInnerBoundary} \\[2mm]
|\nodeboundary{(Q(\overline{V_X})}| 
&\geq& |Q(\overline{V_X})| + 1 \nonumber \\
&\geq& (|S| - 1) - q(V_X) + 1 = |S| - q(V_X).
\label{mathOuterBoundary}
\end{eqnarray}  
By inequalities \ref{mathIntersection}, \ref{mathInnerBoundary},
and \ref{mathOuterBoundary}, we can bound the size of 
$\nodeboundary{Q}$ as follows:
\begin{eqnarray}
|\nodeboundary{Q}| &\geq& |(\nodeboundary{Q}) \cap V_X| + 
|\nodeboundary{Q(\overline{V_X})}| - q(P) \nonumber \\
&\geq& (|S| - k - q(V_X))/2 + |S| - q(V_X) - q(P) \label{mathTotalBoundary1}\\
&\geq& 3(|S| - k - q(V_X))/2 + k - q(P) \label{mathTotalBoundary2}.
\end{eqnarray}

\begin{figure}[t]
\begin{center}
\includegraphics[keepaspectratio,width=110mm]{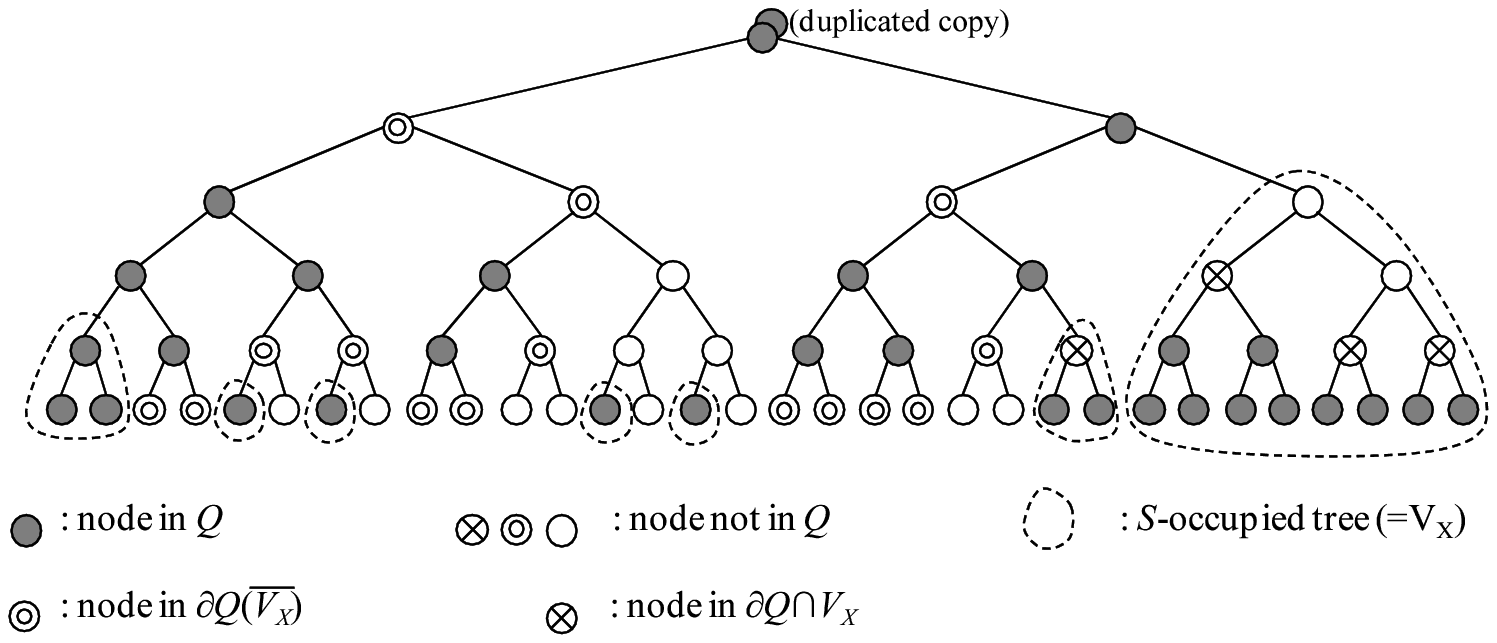}
\caption{Illustration of the set $(\nodeboundary{S})\cap V_X$, 
$\nodeboundary{Q(\overline{V_X})}$, and their boundaries.
in the proof of Theorem~\ref{thmExpansion}.}
\label{fig:tree-Boundary}
\end{center}
\end{figure}

We consider the following two cases according to the value of $k$:

\noindent
{\bf (Case 1)} $k > |S|/16$: We show $|\nodeboundary{Q}| > |S|/240$ holds for 
any $\Pi$. If $q(P) \leq 12k/13$ holds, we have 
$|\nodeboundary{Q}| \geq k/13 \geq |S|/240$ from 
inequality \ref{mathTotalBoundary2} because of $q(V_X) \leq |S| - k$.
Furthermore, if $|S| - q(V_X) - q(P) \geq |S|/240$, we have 
$|\nodeboundary{Q}| \geq |S|/240$ from inequality \ref{mathTotalBoundary1}. 
Thus, in the following argument, we assume $q(P) > 12k/13$ and 
$|S| - q(V_X) - q(P) < |S|/240$.
Consider the subgraph $H$ of $T$ induced by $Q(P)$. We estimate
the number of connected components in $H$ to get a bound 
of $|\nodeboundary{V_H}|$ in $T$. Letting 
$\mathcal{C} = \{C_1, C_2, C_3, \cdots C_m\}$ 
be the set of connected components of $H$, we associate each leaf node 
$u$ in a $S$-occupied subtree $X_i$ with the component in
$\mathcal{C}$ containing $p(X_i)$ (the node is associated 
with no component if $p(X_i)$ is not in $Q(P)$). Since 
each node $v \in H$ has one child belonging to $V_X$, each component
in $H$ forms a line graph monotonically going up to the root. Thus,
if a component $C_i$ has $j$ nodes $u_1, u_2, \cdots u_j$, which 
are numbered from the leaf side, each node $u_h$ ($1 \leq h \leq j$)
has a child as the root of $S$-occupied trees having at least 
$2^{h - 1}$ leaf nodes (recall that $T$ is a complete binary tree). 
It follows that the number of nodes in $S$ associated with $C_i$ is 
$\sum_{h=1}^{j} 2^{h - 1} \geq 2^{j} - 1$. Letting $l_i$ be 
the number of components in $\mathcal{C}$ consisting of $i$ nodes, 
we have:
\begin{eqnarray*}
|S|/m &\geq& \sum_{i=0}^{n}\frac{l_i}{m}\left(2^{i} - 1\right) \\
&\geq& 2^{\sum_{i=0}^{n} i \cdot (l_i/m)} - 1 \\
&\geq& 2^{\frac{12k}{13m}} - 1 \\
&\geq& 2^{\frac{3|S|}{52m}} - 1,
\end{eqnarray*}
where the second line is obtained by applying Jensen's inequality.
To make the above inequality hold, the condition $|S|/m \leq 120
\Leftrightarrow m \geq |S|/120$ is necessary. Next, we calculate 
how many nodes in $\nodeboundary{V_H} \cap \overline{V_X}$ is 
occupied by $Q$. From the definition of $H$, $Q(\nodeboundary{V_H})$ 
does not contain any node in $P$ (if a node $v \in P$ is contained,
it will be a member of $V_H$ and thus not in $\nodeboundary{V_H}$). 
Thus, any node in $\nodeboundary{V_H}$ is a member of $V_X$, 
$\overline{Q} \cap P$, or $\overline{V_X \cup P}$. Let $Y = 
Q(\overline{V_X \cup P} \cap \nodeboundary{V_H})$ and $y = |Y|$ 
for short. Since any node in $\overline{Q} \cap P$ is not contained 
in $Q$ and the cardinality of $\nodeboundary{V_H} \cup V_X$ can
be bounded by $q(P)$ (as the roots of $S$-occupied trees),
we have the following bound from Lemma \ref{lmaGeneralBoundarySize}:
\begin{eqnarray*}
|\nodeboundary{Q}| &\geq&
|\nodeboundary{V_H} \setminus Q| \\
&\geq& |(\nodeboundary{V_H}) \cap (\overline{Q} \cap P)| \\
&\geq& |\nodeboundary{V_H}| - |(\nodeboundary{V_H}) \cap V_X| -
 |(\nodeboundary{V_H}) \cap (\overline{V_X \cup P})| \\
&\geq& |\nodeboundary{V_H}| - q(P)- y \\
&\geq& q(P) + m + 1 - q(P) - y \\
&\geq& m - y.
\end{eqnarray*}

The illustration explaining this inequality is shown in 
Figure \ref{fig:tree-Boundary2}.
Since $Y$, $Q(P)$ and $Q(V_X)$ are mutually disjoint, we obtain 
$y + q(V_X) + q(P) \leq |S| \Leftrightarrow y \leq 
|S| - q(V_X) - q(P) < |S|/240$. 
Consequently, we obtain $|\nodeboundary{Q}| \geq |S|/240$.

\noindent
{\bf (Case 2)} $k \leq \frac{|S|}{16}$:
In the following argument, given a set $S$ satisfying 
$k < |S|/16$, we bound the 
probability $\Pr(|\nodeboundary{Q}| < |S|/32)$. 
From the inequality \ref{mathTotalBoundary2} and the fact of $k - q(P) \geq 0$,
we have $|\nodeboundary{Q}| \geq 3(15|S|/16 - q(V_X))/2$.
To be $|\nodeboundary{Q}| \leq |S|/32$, we need 
$3(15|S|/16 - q(V_X))/2 \leq |S|/32 \Leftrightarrow q(V_X) \geq 11|S|/12$.
Thus, from Lemma \ref{lmaContainmentProb}, we can bound the 
probability as follows:
\begin{eqnarray*}
\Pr(|\nodeboundary{Q}| < |S|/32) &\leq& \Pr(|\Pi(S)| \geq 11|I(V_X)|/12) \\
&\leq& {|S| \choose 11(|S| - k)/12}\left(\frac{|S| - k}{n}\right)^{11|S|/12} \\
&\leq& {|S| \choose (|S| + k)/12}\left(\frac{|S|}{n}\right)^{11|S|/12}.
\end{eqnarray*}
Fixing $k$ and $|S|$, we look at the number of possible choices of $S$. 
Since we can determine a $S$-occupied subtree $X_i$ by choosing one node 
in $T$ as its root, the set $S$ is determined by choosing $k$ nodes 
from all nodes in $T$. Thus, the total number of subset $S$ 
organizing at most $|S|/16$ $S$-occupied subtrees are bounded by 
$\sum_{i = 1}^{|S|/16} {2n \choose i} \leq 
\frac{|S|}{16} {2n \choose |S|/16}$. Summing up this bound for any 
$|S| < n/2$. The total number is bounded by $\sum_{|S| = 1}^{n/2}
\frac{|S|}{16}{2n \choose |S|/16}$. 
Using the union bound and a well-known bound ${n \choose m} 
\leq (ne/m)^m$, we have:
\begin{eqnarray*}
\lefteqn{\Pr\left(\bigcup_{S \subseteq L(V_T) | |S| \leq n/2} 
|\nodeboundary{Q}| < \frac{|S|}{32}\right)} \hspace{5mm}\\
&\leq& \sum_{|S| = 1}^{n/2} \frac{|S|}{16} 
{2n \choose {|S|/16}}{|S| \choose (|S| + k)/12}
\left(\frac{|S|}{n}\right)^{11|S|/12}. \\
&= & o(1).
\end{eqnarray*}
All the details of the previous calculation are provided in the Appendix.
The theorem is proved. \qed
\end{proof}

\begin{figure}[t]
\begin{center}
\includegraphics[keepaspectratio,width=110mm]{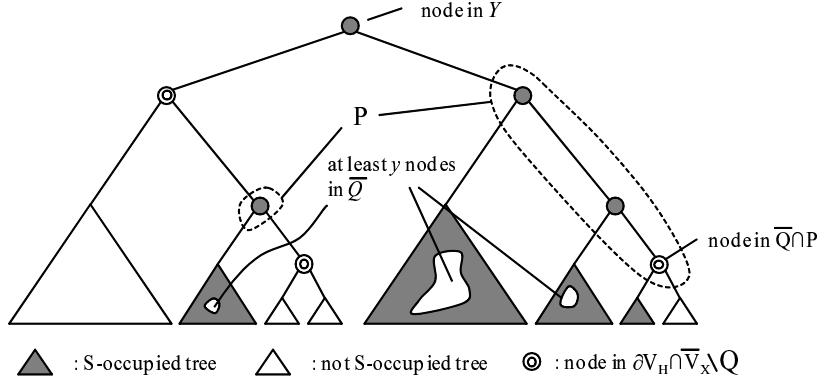}
\caption{Illustration of the boundary 
$(\overline{V_H} \cap \overline{V_X}) \setminus$ in the proof of 
Theorem~\ref{thmExpansion}.}
\label{fig:tree-Boundary2}
\end{center}
\end{figure}

%% file: mixing.tex
\section{Distributed Construction of $G_\Pi$}
\label{secDistributed}

To prove the impact of Theorem \ref{thmExpansion} in P2P settings 
we have to construct scalably a random bijection (i.e., random 
perfect bipartite matching) between internal and leaf nodes in tree overlays. 
In this section, we show that this distributed construction is possible 
with nice self-$\ast$ properties. That is, our scheme is 
totally-distributed, uses only local information, and is self-healing 
in the event of nodes joins and leaves. In the following we state the
computational model and our network assumptions.

\subsection{Computational Model}
We consider a set of virtual nodes (peers) distributed over a 
connected physical network. Virtual nodes are structured in a  binary tree overlay. 
Each physical node managing a virtual node $v$
can communicate with any physical node managing $v$'s neighbors in the 
overlay. The communication is synchronous and round-based. That is,
the execution is divided into a sequence of consecutive rounds. All messages sent
in some round are guaranteed to be received within the same round.

We assume  that each physical node
manages exactly one internal node and one leaf node in the virtual overlay. Moreover, we also assume that 
the tree is balanced. Note that these assumptions are not far from 
practice. Most of distributed tree overlay implementations embed balancing 
schemes. The preservation of matching structure is easily guaranteed 
by employing the strategy that one physical node always join as two new 
nodes. We can refer as an example the join/leave algorithm in 
\cite{MLR07}, which is based on the above strategy and generally 
applicable to most of binary-tree overlays.

\subsection{Uniformly-Random Matching Construction}
The way leaf and internal nodes are matched via a physical node 
is generally dependent on the application requirements and is rarely chosen uniformly at random.
That is, the implicit matching offered by the overlay
may be extremely biased 
and the expansion factor computed in the previous section may not 
hold. Fortunately, the initial matching can be quickly ``mixed'' to 
obtain a uniformly-random matching. To this end we will extend 
the technique proposed by Czumaj et. al.\cite{CK00} for fast random 
permutation construction to distributed scalable matchings 
in tree overlays. The following stochastic process rapidly mixes the 
sample space of all bipartite matchings between leafs and internal nodes:
\begin{enumerate}
\item Each leaf node first tosses a fair coin and decides whether 
it is active or passive.
\item Each active node randomly proves a leaf node and sends a 
matching-exchange request.
\item The passive node receiving exactly one matching-exchange request
accepts it, and establishes the agreement to the sender of the request.
\item The internal nodes managed by the agreed pair are swapped.
\end{enumerate}

Note that the above process is performed by all leaf nodes concurrently 
in a infinite loop. Following the analysis by Czumaj et. al.\cite{CK00}, 
the mixing time of the above process is $O(\log n)$. 

The only point that may create problems in distributed P2P settings is 
the second step. For that point, we can propose a simple solution to 
implement the random sampling mechanism with $O(\log n)$ time and 
message complexity. The algorithm is as follows:  First, the prober 
sends a token to the root. From the root, the token goes down along the 
tree edges by selecting with equal probability one of its children. 
When the token reaches a leaf node, the destination is returned as 
the probe result. 

Overall, the distributed scalable algorithm for constructing a 
random bipartite matching takes $O(\log^2 n)$ time. In the following 
subsection, we evaluate the performances of the above scheme 
face to churn.

%% file: simulations.tex
\subsection{Experimental Evaluation in Dynamic Environment}

In this section,
we experimentally validate the performance of our approach by
simulation. In the simulation scenario the following four phases
are repeated.

\begin{description}
\item[Node join] We assume that a newly-joining node knows 
the physical address of some leaf node $u$ in the network. 
Let $v$ and $v'$ be the leaf and internal nodes that will be managed 
by the newly-joining physical node.  The node $u$ is replaced by 
a newly internal node $v'$. Then $v$ and $u$ becomes children of $v'$.
\item[Node leave] The adversary chooses a number of nodes to make them
leave. Since it is hard to simulate worst-case adversarial behavior, we
adopt a heuristic strategy: given a physical node $v$ with leaf $v_L$ and 
internal node $v_I$, let $h(v)$ be the height of the smallest subtree
containing both $v_L$ and $V_I$. Intuitively, the physical node $v$
with higher $h(v)$ has much contributition for avoiding the node boundary 
to be contained in a small subtree containing $v_L$. Following this 
intuition, the adversary always makes the node $v$ with highest $h(v)$
leave.
\item[Balancing] Most of tree-based overlay algorithms have some 
balancing mechanism. While the balancing mechanism has a number of
variations, we simply assume a standard rotation mechanism.
After a number of node joins and leaves, the tree is balanced by 
standard rotation operation. 
\item[Matching Reconstruction] We run once the matching-mixing process described in the previous section.
\end{description}

Since exact computation of node expansion 
is coNP-complete, we monitor the second smallest eigenvalue 
$\lambda$ of the graph's Laplacian matrix, which has a strong corelation 
to the node expansion: a graph with the second smalles eigenvalue 
$\lambda$ is a $\lambda/2$-expander.
In the following we propose our simulations results first in a churn free setting then in environments with different churn levels. Due to the space limitation the churn-free simulations are defered to the Anexe section.

\paragraph{Without churn} We ran 100 simulations of 100 rounds with 512 nodes and no churn.
The value of $\lambda$ is calculated at the begin of each round.
Those simulations tends to emphasize what is ``expectable'' from the mixing protocol and some of its dynamic properties.

$\lambda$ varies from 0.263 to 0.524 with an average value of 0.502 and a standard deviation of 0.017.
The low standard deviation and the closeness of average and maximum reached values of $\lambda$ indicates that the minimum is rarely reached.
Basically it is obtained when most nodes become responsible of internal nodes that are ``close'' from their leaves.
Intuitively if each node is responsible of an ancestor of its leaf, there is no additionnal links between the left and the right subtrees of the root.
In that case we do not take benefit of mixing and get bad expansion properties inherited from tree structures.

\paragraph{Churn prone environments.} We ran 100 simulations of 100 rounds with 512 nodes and a given rate of churn.
Time is divided in seven rounds groups.
During the first round a given percentage of new nodes join the system.
During the second a given percentage of nodes leave the system.
During the third round the tree is balanced and the mixing protocol is run.
During other rounds the mixing protocol is run.
$\lambda$ is measured at the end of each round.
Nodes gracefully leave the system.
Those simulations investigate the impact of churn on $\lambda$ and how fast our mixing protocol restores a stable configuration.

\begin{figure}
\center
\includegraphics[angle=-90,width=.5\textwidth]{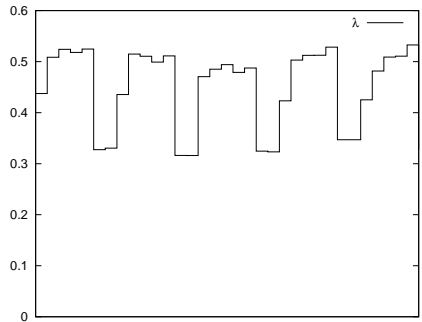} 
\caption{$\lambda$ over time with 10\% of churn}
\label{fig:l-churn10}
\end{figure}

\begin{figure}
\center
\includegraphics[angle=-90,width=.5\textwidth]{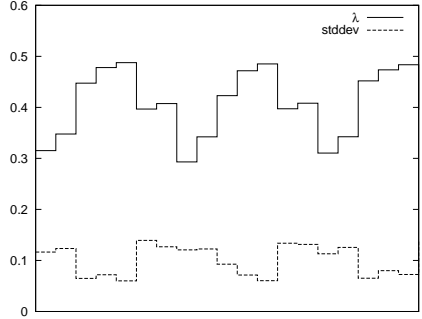}
\caption{$\lambda$ over time with 30\% of churn}
\label{fig:l-churn30}
\end{figure}

Figure~\ref{fig:l-churn10} shows the evolution of $\lambda$ over time in presence of 10\% of churn (10\% is relative to the initial number of nodes).
Each step stands for a round.
From a stable configuration where $\lambda$ oscillates between 0.52 and 0.48, it drops down to 0.4 every seven rounds due to arrivals and departures.
The structure is sensitive to churn in the sense that it significantly decreases the value of $\lambda$.
But on the other hand, proposed mixing protocol converges fast.
It needs two rounds to reach the average ``expectable'' value of $\lambda$.

Figure~\ref{fig:l-churn30} shows the evolution of $\lambda$ over time in presence of 30\% of churn (30\% is relative to the initial number of nodes).
Each step stands for a round.
Basically churn increases stretches the curve; down part are downer around 0.3, high values are quite stable around 0.45 and fluctuations are widespread.
While the previous one~\ref{fig:l-churn10} gives a good overview of the global behaviour of the protocol facing churn, this figure emphases some interesting details.
First it emphasizes that mixing protocol is not monotonic; it might decrease $\lambda$.
Second, the impact of arrivals and departures are distinct.
Moreover the magnitude of their impact is not predictable because the selection of bootstrap node is random. 
Starting from a stable situation arrivals will always decrease $\lambda$.
With the proposed join mechanism, a new comer is weakly connected to the rest of the system.
Starting from a stable situation gracefull departures will almost always decrease $\lambda$.
But with the proposed leave mechanism their impact is more subtle because; they can largely modify the tree balance which also implies links exchanges.
In some rare cases those exchanges (which could be thought as side effect shuffles of links) or the departure of weakly connected nodes could increase $\lambda$.

%% file: conclusion.tex
\section{Concluding Remarks}
\label{secConclusion}

We proposed for the first time in the context of overlay networks a generic scheme that transforms any tree overlay in an expander with constant 
node expansion with high probability. More precisely, we prove that 
a uniform random tree virtualization yields a node expansion at 
least $\frac{1}{480}$ with probability $1 - o(1)$.  Second, 
in order to demonstrate the effectiveness of our result in the context 
of real P2P networks we further propose and evaluate in different 
churn scenario a simple scheme for uniformly random tree 
virtualization in $O(log^2n)$ running time. Our scheme is totally 
distributed and uses only local information. Moreover, in the event 
of nodes join/leave or crash our scheme is self-healing.

The virtualization scheme itself is a promising approach and provides
several interesting questions. We enumerate the open problems related 
to our result:
\begin{itemize}
\item {\bf Better analysis of matching convegence}: As the 
simulation result exhibits, the convergent value of expansion computed 
from $\lambda$ is considerably larger than the theoretical bound. In addition,
the convegence time is also faster than the theoretical bound. Finding
improved bounds for both of the expansion and convegence time is an
open problem.
\item {\bf Effective utilization of expander propoerty}: In addition to
fault resiliency, expander graphs also offer the rapidly-mixing property 
of random walks on the graph. That is, MCMC-like sampling method 
effectively runs on our scheme. It is an interesting research direction 
that we utilize  expansion property for implementing some statistical 
operation over distributed datas or query load balancing.
\item {\bf Application of virtualization scheme to other overlays}:
The virtualization scheme is simple and generic, and thus we can apply
it to other well-known overlay algorithms such as Chord or Pastry. 
Clarifying the class of overlay networks where the virtualizaton scheme
efficiently works is a challenging problem.
\end{itemize}

%% file: appendix.tex
\section{Omitted proofs}

We explain the proof details omitted in the main body.

\subsection{Proof of Lemma \ref{lmaConnectedBoundarySize}}
\begin{proof}
We prove the case of $r(T) \not\in S$ by 
induction on the cardinality of $S$. 
({\bf Basis}) If $|S| = 1$, the lemma clearly holds because
every internal node except for the root of $T$ has degree three.
({\bf Inductive step}) Suppose as the induction hypothesis
that $|\nodeboundary{S}| \geq |S| + 2$ holds for any connected set 
$S$ of size $k$. we prove that for any $v \in I(V_T)$ that 
is connected to a node in $S$, $|\nodeboundary{(S \cup \{v\})}| 
\geq |S \cup \{v\}| + 2$ holds. For short, let $S' = S \cup \{v\}$. 
Since $S$ is connected, $v$ is connected 
to exactly one node in $S$. In other words, node $v$ has two neighbors of 
$v$ in neither $S$ nor $\nodeboundary{S}$, which are elements 
of $\nodeboundary{S'}$. In contrast, $v$ is an element of 
$\nodeboundary{S}$ but not in $\nodeboundary{S'}$. It follows that
$\nodeboundary{S'} \geq \nodeboundary{S} - 1 + 2$ holds. From the 
induction hypothesis, we obtain $\nodeboundary{S'} \geq 
\nodeboundary{S} + 1 \geq |S| + 1 + 2 = |S'| +2$.
The case $r(T) \in S$ is obviously deduced from the case of $r(T)
\not\in S$. The lemma is proved. \qed
\end{proof}

\subsection{Proof of Lemma \ref{lmaGeneralBoundarySize}}
\begin{proof}
Let $C_1, C_2, \cdots C_j, \cdots C_m$ be the set of connected components 
in $\mathrm{Ind}(S)$. In the case of $r(T) \in S$, we assume 
$r(T) \in C_1$ without loss of generality. We prove the lemma by induction
on $j$. ({\bf Basis}) It holds from Lemma \ref{lmaConnectedBoundarySize}
({\bf Inductive step}) 
Suppose $|\nodeboundary{S}| \geq |S| + j + 1$ as the induction hypothesis.
Consider adding a new component $C_{j+1}$ into $S$. Let $c$ be the 
number of nodes in $C_{j+1}$. Since $r(T) \not\in C_{j+1}$, from 
Lemma \ref{lmaConnectedBoundarySize}, $|\nodeboundary{V_{C_{m+1}}}| 
\geq c + 2$ holds. At most one node is shared by 
$\nodeboundary{S}$ and $\nodeboundary{V_{C_{m+1}}}$, we have 
$|\nodeboundary{S \cup V_{C_{j+1}}}| \geq |S| + j + 1 + c + 2 - 1 
\geq (|S| + c) + (j + 1)$. The lemma is proved. \qed
\end{proof}

\subsection{Proof of Lemma \ref{lmaInternalBoundary}}
\begin{proof}
We omit the subscript $\Pi$ of $Q_\Pi$ for short.
We divides $\overline{Q} \cap V_{X}$ into three mutually-disjoint
subset $\overline{S_1}$, $\overline{S_2}$, and 
$\overline{S_3}$: Let $\overline{S_1} \subseteq V_X \cap \overline{Q}$ 
be the set of nodes that have no neighbor belonging to $Q \cap V_{X}$,
$\overline{S_2} = \nodeboundary{\overline{S_1}} \cap \overline{Q}$, and 
$\overline{S_3} = (\overline{Q} \cap V_{X}) \setminus (\overline{S_1} \cup 
\overline{S_2})$ (see Figure \ref{fig:tree-S1S2S3}). 
Since $X$ is $S$-occupied, $\overline{S_2}$ 
consists only of internal nodes. Thus, from Corollary 
\ref{corolBoundarySize}, $|\overline{S_2}| = 
|\nodeboundary{\overline{S_1}}| \geq |\overline{S_1}|$ holds. 
By the definition of $\overline{S_2}$ and $\overline{S_3}$, 
$\overline{S_2} \subseteq (\nodeboundary{Q}) \cap V_{X}$ and 
$\overline{S_3} \subseteq (\nodeboundary{Q}) \cap V_{X}$ hold. 
Consequently, we have  
\begin{eqnarray*}
2|\nodeboundary{Q} \cap V_{X}| 
&\geq& 2(|\overline{S_2}| + |\overline{S_3}|) \\
&\geq& |\overline{S_2}| + 2|\overline{S_3}| + |\overline{S_1}|\\
&\geq& |\overline{Q} \cap V_{X}|.
\end{eqnarray*}
The lemma is proved. \qed
\end{proof}

\begin{figure}[t]
\begin{center}
\includegraphics[keepaspectratio,width=110mm]{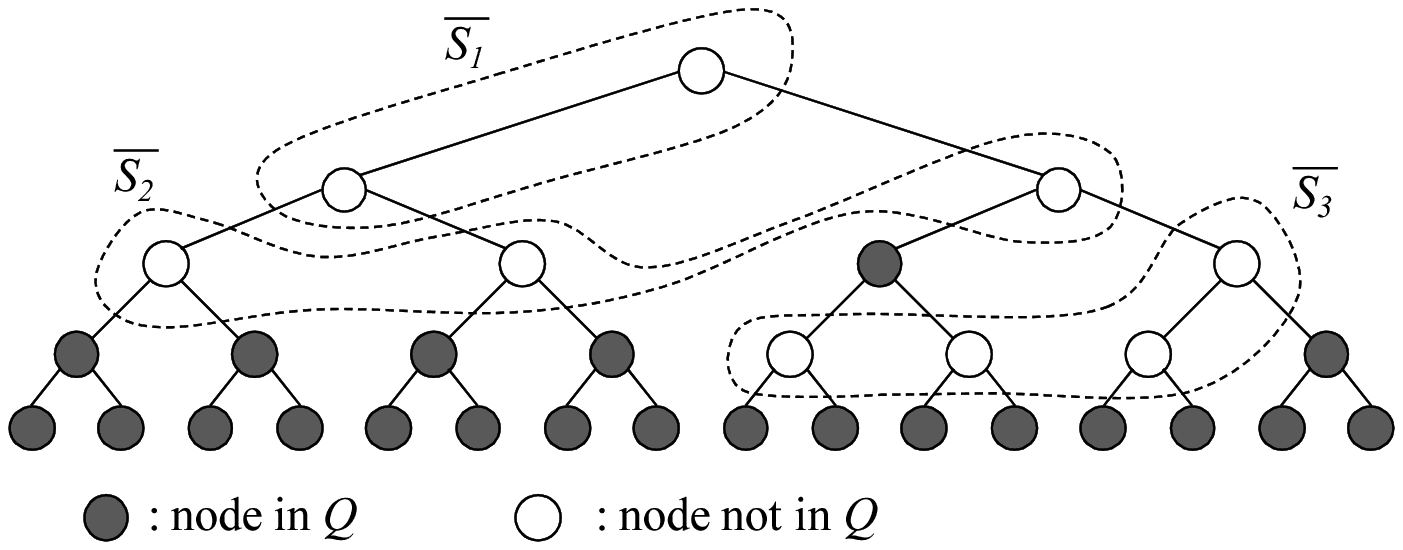}
\caption{Illustration of $\overline{S_1}$, $\overline{S_2}$, and
$\overline{S_3}$ in the proof of Lemma~\ref{lmaInternalBoundary}.}
\label{fig:tree-S1S2S3}
\end{center}
\end{figure}

\subsection{Proof of Lemma \ref{lmaContainmentProb}}

\begin{proof}
Let $Z = I(V_X)$ for short.
First, we fix a subset $S' \subseteq$ of $S$ with cardinality 
$\alpha|Z|$, and consider the probability that 
$\Pi(S') \subseteq Z$ holds.
Without loss of generality, we can regard $\Pi$ as a permutation on 
$Z$ whose first $|S|$ elements are mapped to $Z$. Thus, 
to compute the probability, it is sufficient to count the number 
of permutations where any element in $S'$ appears at the first 
$|Z|$ elements of $\Pi$. Dividing the counted number by 
$n!$, the probability can be calculated as follows:
\begin{eqnarray*}
\Pr(S' \subseteq X) &=&
\frac{1}{n!} {{n - |S'|}  \choose {|Z| - |S'|}} |Z|! 
(n - |Z|)! \\
&=& \frac{|Z|(|Z| - 1)(|Z| - 2) \cdots (|Z| - \alpha|Z| + 1)}{
n(n - 1) (n - 2) \cdots (n - \alpha|Z| + 1)} \leq 
\left(\frac{|Z|}{n}\right)^{\alpha|Z|}.
\end{eqnarray*}

Using the union bound, we can obtain the following bound:
\begin{eqnarray*}
\Pr(\Pi(S) \geq \alpha|Z|) &\leq&
\Pr(\bigcup_{S' \subset S | S' = \alpha|Z|} S' \subseteq Z) \\
&\leq& {|S| \choose \alpha|Z|}\left(\frac{|Z|}{n}\right)^{\alpha|Z|}.
\end{eqnarray*}
Since $|Z| = |S| - k$ holds, the lemma is proved. \qed
\end{proof}

\subsection{The Calculation Details in the Proof of Theorem \ref{thmExpansion}}

We describe the detailed calculation to lead the last inequality in 
the proof of Theorem \ref{thmExpansion}.
\begin{align*}
\lefteqn{\Pr\left(\bigcup_{S \subseteq L(V_T) | |S| \leq n/2} 
|\nodeboundary{Q}| < \frac{|S|}{32}\right)} \hspace{5mm}\\
&\leq \sum_{|S| = 1}^{n/2} \frac{|S|}{16} 
{2n \choose {|S|/16}}{|S| \choose (|S| + k)/12}
\left(\frac{|S|}{n}\right)^{11|S|/12}. \\
\intertext{Using the bound ${n \choose m} 
\leq (ne/m)^m$ and the condition $k < |S|/16$,} \\
&\leq \sum_{|S| = 1}^{n/2} 
\frac{|S|}{16}\left(\frac{32en}{|S|}\right)^{|S|/16}
\left(\frac{192e}{17} \right)^{17|S|/192}
\left(\frac{|S|}{n}\right)^{11|S|/12} \\
&\leq \sum_{|S| = 1}^{n/2}
\frac{|S|}{16}\left((32e)^{1/16}(192e/17)^{17/192}\right)^{|S|}
\left(\frac{|S|}{n}\right)^{(11/12 - 1/16)|S|} \\
\intertext{By numeric calculation, we have
$\log ((32e)^{1/16}(192e/17)^{17/192}) \leq 0.841$ and 
$(11/12 - 1/16) \geq 0.854$. Thus,} \\
&\leq \sum_{|S| = 1}^{n/2}
\frac{|S|}{16} 2^{0.841|S|} \left(\frac{|S|}{n}\right)^{0.854|S|} \\
&=  o(1).
\end{align*}